\documentclass[a4paper]{lipics-v2021}

\nolinenumbers

\usepackage{graphicx,amsthm,amsmath,amssymb,url,hyperref,bbm}

\title{Equivalence of Continuous-Time Markov Chains and Linear Dynamical Systems}

\author{Mihir Vahanwala}{MPI-SWS, Saarland Informatics Campus }{mvahanwa@mpi-sws.org}{https://orcid.org/0009-0008-5709-899X}{}

\authorrunning{Mihir Vahanwala}

\Copyright{Mihir Vahanwala}



\begin{document}

\maketitle
\begin{abstract}
    The purpose of this short note is to record that an analogue of the following result, which is known for discrete-time linear dynamical systems, also holds in the continuous-time setting. The dynamics of a $d$-state Markov chain is governed by that of a linear dynamical system of dimension at most $d-1$; conversely, a linear dynamical system of dimension $d-1$ can be ``embedded'' into a Markov chain with $d$ states.
\end{abstract}

\section{Statements}

A discrete-time linear dynamical system of dimension $d$ is given by an update matrix $M \in \mathbb{Q}^{d \times d}$ and an initial vector $u_0 = u \in \mathbb{Q}^{d}$. The system evolves as $u_{n+1} = M u_n$, and hence $u_n = M^n u_0$. We say that this system is a (discrete) Markov chain if (i) the initial vector $u_0$ is a distribution, i.e., all entries are non-negative, and the sum of entries is $1$, and (ii) each column of $M$ is a distribution, making it a stochastic matrix. We immediately observe that the orbit is thus a sequence of distributions. 

Analogously, a continuous-time linear dynamical system of dimension $d$ is given by an infinitesimal generator (or simply generator) matrix $A \in \mathbb{Q}^{d \times d}$ and an initial vector $v(0) \in \mathbb{Q}^d$. The system evolves as $\frac{d}{dt}v(t) = A v(t)$, and hence $v(t) = \exp(At)v(0)$. We say that this system is a continuous Markov chain if (i) the initial vector $v(0)$ is a distribution as before, and (ii) the entries in each column of the generator matrix $A$ add up to $0$, and the off-diagonal entries are non-negative. We refer to such a matrix $A$ as a stochastic generator matrix.

We prove analogues of the core results in \cite{Aghamov2025} (see also \cite{Vahanwala}). For notational convenience, we use $\mathbf{0}_d$ to denote the $d$-dimensional vector with all entries $0$, $\mathbf{1}_d$ to denote the $d$-dimensional vector with all entries $1$, and $I_d$ to denote the $d \times d$ identity matrix. We may omit the subscripts when the dimensions are clear from the context. We use $O$ to denote a block of appropriate dimensions, all of whose entries are $0$. The symbol $\top$, when in superscript, denotes matrix transposition. Linear algebraic notions that we assume familiarity with can be found in \cite[Sec.~2.3]{Aghamov2025}. 

\begin{theorem}
\label{thm::stochastic-to-regular}
    Given a distribution $v \in \mathbb{Q}^d$ and a stochastic generator matrix $A \in \mathbb{Q}^{d \times d}$, let $\ell \ge 1$ be the multiplicity of the eigenvalue $0$ of $A$. We can compute a generator matrix $B \in \mathbb{Q}^{(d - \ell) \times (d - \ell)}$, an embedding matrix $Q \in \mathbb{Q}^{d \times (d -\ell)}$, a distribution $s \in \mathbb{Q}^d$, and a vector $u \in \mathbb{Q}^{d-\ell}$ such that for all $t \ge 0$, we have that 
    $$
    \exp(At) v = s + Q \exp(Bt) u.
    $$
    In the above, the choice of $Q$ is independent of $v$.
\end{theorem}

\begin{theorem}
\label{thm::regular-to-stochastic}
    Given a vector $u \in \mathbb{Q}^{d-1}$ and a generator matrix $C \in \mathbb{Q}^{(d-1)\times (d-1)}$, we can compute $\rho, \eta \in \mathbb{Q}$, a distribution $v \in \mathbb{Q}^d$, and a stochastic generator matrix $A \in \mathbb{Q}^{d \times d}$ such that for all $t \ge 0$ we have that
    $$
    \exp(At) v = s + \eta e^{-\rho t} \cdot Q \exp(Ct) u,
    $$
    where the pair $s, Q$ is arbitrarily chosen as $s = \mathbf{1}_d/d$ and $Q = \begin{bmatrix}
        I_{d-1} \\
        - \mathbf{1}_{d-1}^\top
    \end{bmatrix}$.
\end{theorem}

\section{Proofs}
We begin by noting a basic property of generator matrices, which resolves a potential ambiguity in the statement of Thm.~\ref{thm::stochastic-to-regular} regarding the multiplicity of the eigenvalue $0$.

\begin{lemma}
    \label{lem::stochastic-generator-spectrum}
    Let $A$ be a stochastic generator matrix. We have that $A$ has $0$ as an eigenvalue, and its algebraic multiplicity is equal to its geometric multiplicity. Furthermore, all other eigenvalues of $A$ have negative real part.
\end{lemma}
\begin{proof}
    It follows by definition that $\mathbf{1}^\top A = \mathbf{0}^\top$, and we thus have found a left eigenvector for the eigenvalue $0$. Observe that we can choose $a > 0$ small enough so that $M = I + aA$ is a stochastic matrix. It is easy to check that $\lambda$ is an eigenvalue of $A$ if and only if $\lambda'= 1 + a\lambda$ is an eigenvalue of $M$, and that their corresponding eigenvectors have the same order, and span the same spaces. Perron-Frobenius theory (see, e.g., \cite[Lem.~3]{Aghamov2025}) gives us the desired properties of the spectrum of $M$: it necessarily has $1$ as an eigenvalue with equal algebraic and geometric multiplicities, and all other eigenvalues have absolute value at most $1$. The former corresponds to the eigenvalue $0$ of $A$, the latter can only correspond to eigenvalues of $A$ that are strictly to the left of the imaginary axis.
\end{proof}

The proofs of Thms.~\ref{thm::stochastic-to-regular} and \ref{thm::regular-to-stochastic} mirror the proofs in \cite[Sec.~3]{Aghamov2025}. 

\begin{proof}[Proof of Thm.~\ref{thm::stochastic-to-regular}]
Let $V_0$ be the $\ell$-dimensional subspace spanned by the eigenvectors of the eigenvalue $0$, and let $W$ be the $(d -\ell)$-dimensional subspace (of $\mathbb{C}^d$) spanned by the (generalised) eigenvectors of all the other eigenvalues. Recall the basic properties that $V_0, W$ together span $\mathbb{C}^d$ \cite[Thm.~6]{Aghamov2025}, and that $W$ is perpendicular to the space spanned by the left eigenvectors of the eigenvalue $0$.

We can thus compute matrices $P \in \mathbb{Q}^{d \times \ell}, Q \in \mathbb{Q}^{d \times (d-\ell)}$ whose columns respectively form bases of $V_0, W$. Observe that $v$ can be uniquely expressed as $Pw + Qu$, where $w \in \mathbb{Q}^\ell$, and $u \in \mathbb{Q}^{d-\ell}$. We choose $s = Pw$.

It remains to define $B$, and then show that $\exp(At)v = s + Q \exp(Bt) u$. To that end, define the invertible $R = \begin{bmatrix}
    P & Q
\end{bmatrix}$, and express $R^{-1}$ as $\begin{bmatrix}P' \\ Q' \end{bmatrix}$. Consider the matrix
\begin{align*}
    D &= R^{-1}AR \\
    &= \begin{bmatrix}P' \\ Q' \end{bmatrix} A \begin{bmatrix} P & Q \end{bmatrix} \\
    &= \begin{bmatrix}P' \\ Q' \end{bmatrix} \begin{bmatrix} O & AQ \end{bmatrix} \\
    &= \begin{bmatrix}
        O & O \\
        O & Q'AQ.
    \end{bmatrix}
\end{align*}
In the above, we used the property that all of the columns of $AQ$ are in $W$, and that the rows of $P'$ span the space perpendicular to $W$. We define $B = Q'AQ$, and observe that by similar reasoning as above, for all $n \ge 1$
$$
D^n = \begin{bmatrix}
    O & O \\ O & B^n
\end{bmatrix} = \begin{bmatrix}
    O & O \\ O & Q'A^nQ
\end{bmatrix}.
$$ Dually, we also deduce that for all $n \ge 1$, we have $A^n= QB^n Q'$.

Let us now express the matrix exponential using the formal power series. We get
\begin{align*}
    \exp(At) v &= R \left( \exp(Dt) \right) R^{-1}v \\
    &= \begin{bmatrix}
        P & Q
    \end{bmatrix}
    \left(\sum_{i=0}^\infty \frac{t^i}{i!} D^i \right)
    \begin{bmatrix}
        w \\ u
    \end{bmatrix} \\
    &= s + Qu + \sum_{i=1}^\infty \frac{t^i}{i!}
    \begin{bmatrix}
        P & Q
    \end{bmatrix}
    \begin{bmatrix}
        O & O \\
        O & B^i
    \end{bmatrix}
    \begin{bmatrix}
        w \\ u
    \end{bmatrix} \\
    &= s + Qu + \sum_{i=1}^\infty \frac{t^i}{i!}QB^i u \\
    &= s + Q \left(\sum_{i=0}^\infty \frac{t^i}{i!} B^i \right)u \\
    &= s + Q \exp(Bt) u,
\end{align*}
as desired.
\end{proof}

We remark that \cite[Sec.~3]{Aghamov2025} also discusses how the spectral decompositions of $B$ and $A$ are related. The same discussion applies in the present setting too. In particular, $B$ and $A$ have the same nonzero eigenvalues, and hence $Q \exp(Bt) u$ converges to $\mathbf{0}$.

\begin{proof}[Proof of Thm.~\ref{thm::regular-to-stochastic}]
    We prove the theorem by using the calculations from the previous proof. We have fixed $Q$, and $P$ is the matrix whose single column is $s$. This determines $P', Q'$; in particular $P' = \mathbf{1}^\top$. We choose the scalar $\eta$ to ensure that $v = s + \eta Q u$ is a distribution. The scalar $\rho$ is chosen to define a matrix $B = C - \rho I$ such that $QBQ'$ is a stochastic generator matrix. It would then follow that 
    $$
    \exp(At)v = s + \eta Q \exp(Bt) u = s + \eta e^{-\rho t} \exp(Ct) u.
    $$

    Let us show how to choose $\rho$. We already have that $\mathbf{1}^\top Q C Q' = \mathbf{0}^\top$ because $\mathbf{1}^\top Q = \mathbf{0}^\top$. We only need to show that we can make the off-diagonal entries non-negative by subtracting $\rho I$ from $C$.

    We have
    $$
    Q(C - \rho I)Q' = QCQ' - \rho QQ'.
    $$
    Since $$\begin{bmatrix}
        s & Q
    \end{bmatrix}\begin{bmatrix}
        \mathbf{1}^\top \\ Q'
    \end{bmatrix} = I,$$
    we have that $QQ' = I - S$, where $S$ is the matrix with all entries equal to $1/d$. Thus, $Q(C - \rho I)Q' = QCQ' + \rho S - \rho I$. Choosing $\rho$ to be large enough thus ensures that all the off-diagonal entries of $QBQ'$ are non-negative, as desired.
\end{proof}

\bibliographystyle{plainurl}
\bibliography{main}

\end{document}